\documentclass[11pt]{article}
\usepackage[english]{babel}
\usepackage{amsmath,amssymb,amsthm}

\usepackage[pdfpagemode=UseOutlines, bookmarksnumbered=true, unicode=true,
pdffitwindow=true,pdfpagelayout=SinglePage,pdfhighlight=/P,
colorlinks=true,linkcolor=blue,citecolor=blue,urlcolor=blue]{hyperref}

\usepackage[left=2.5cm, right=2.5cm, top=2.5cm, bottom=2.5cm]{geometry}

\allowdisplaybreaks[4] 

\theoremstyle{plain}
\newtheorem{proposition}{Proposition}
\newtheorem{lemma}[proposition]{Lemma}
\newtheorem{theorem}[proposition]{Theorem}

\theoremstyle{definition}
\newtheorem{remark}[proposition]{Remark}
\newtheorem{example}[proposition]{Example}

\title{Vector systems of Painlev\'e type}

\author{V.E.\:Adler\thanks{L.D.\:Landau Institute for Theoretical Physics, Chernogolovka, Russian Federation},\quad 
V.V.\:Sokolov\thanks{Higher School of Modern Mathematics MIPT, Moscow, Russia} \thanks{National Research University ``Higher School of Economics''. E-mail: vsokolov@landau.ac.ru}} 

\date{January 07, 2026}

\begin{document}
\maketitle

\begin{abstract}
The group reduction procedure is applied to vector generalizations of the NLS, mKdV, and KdV equations. The resulting ODE systems admit isomonodromic Lax representations and are multicomponent generalizations of the Painlev\'e equations \ref*{P1}, \ref*{P2}, \ref*{P34}, and \ref*{P4}. Some of them can be interpreted as nonautonomous deformations of well-known systems integrable in the Liouville sense, in particular, the Garnier and H\'enon--Heiles systems. In one case, an unexpected connection with the equations of quasiperiodic dressing chain for the Schr\"odinger operator is established.
\end{abstract}


\section{Introduction}

It is well known that Painlev\'e equations arise from integrable PDEs when passing to solutions that are invariant with respect to a certain group of point transformations (see, for example, \cite{ConteMusette} and references therein). In particular, the most fundamental models, the Korteweg--de Vries equation (KdV) 
\[
 u_t=u_{xxx}+6uu_x,   
\]
its modification (mKdV) 
\[
 u_t=u_{xxx}+6u^2u_x,
\]
and nonlinear Schr\"odinger equation (NLS)
\[
 u_t=u_{xx}+2u^2v,\quad -v_t=v_{xx}+2v^2u,
\]
are related by similarity reductions with the Painlev\'e equations
\begin{align}
\label{P1}\tag{P$_1$}
 w''&= 6w^2+z,\\
\label{P2}\tag{P$_2$}
 w''&=2w^3+zw+\alpha,\\
\label{P34}\tag{P$_{34}$}
 w''&=\frac{(w')^2}{2w}+2w^2-zw+\frac{\alpha}{w},\\
\label{P4}\tag{P$_4$}
 w''&=\frac{(w')^2}{2w}+\frac{3}{2}w^3+4zw^2+2(z^2-\alpha)w+\frac{\beta}{w},
\end{align}
as shown in the following table:
\begin{center}
\begin{tabular}{c|cc}
    & Galilean & scaling \\
    & transformation & \\
\hline
KdV  & \ref*{P1} & \ref*{P34} \\
mKdV & ---       & \ref*{P2} \\
NLS  & \ref*{P2} & \ref*{P4} \\
\end{tabular}
\end{center}
The goal of our work is to generalize these results for vector analogs of the equations listed above. By vector equations we mean multicomponent systems in which the field variables are vectors of arbitrary dimension which enter the equations linearly with coefficients expressed by functions of scalar products of these vectors and their derivatives. A typical example is the Manakov system \cite{Manakov_1973}, which defines a ${\rm GL}(n)$-invariant generalization of the NLS equation:
\[
  {\bf u}_t={\bf u}_{xx}+2({\bf u},{\bf v}){\bf u},\quad 
 -{\bf v}_t={\bf v}_{xx}+2({\bf u},{\bf v}){\bf v},\quad 
  {\bf u},{\bf v}\in{\mathbb R}^n.
\]
From here on, vectors are denoted in bold. Many vector systems with integrable properties are known (see, for example, \cite{Sokolov_Svinolupov_1994, Sokolov-Wolf, Meshkov-Sokolov}), but in this paper we restrict ourselves to few most well-known ones. In the vector case, the similarity reduction procedure itself does not change essentially, but several peculiarities should be taken into account.

1) For each of the evolution equations above, there are at least two integrable vector analogs; this adds variety to the picture. For the NLS equation, in addition to the Manakov system, there is also the Kulish--Sklyanin system \cite{Kulish_Sklyanin_1981} (which is ${\rm O}(n)$-invariant). As in the scalar case, the third-order symmetries of these two systems admit the reductions ${\bf v}={\bf u}$ and ${\bf v}=\operatorname{const}$, which yield two vector analogs for the mKdV and KdV equations.

It should be noted that even more general NLS and mKdV systems associated with symmetric spaces were introduced in \cite{Fordy_Kulish_1983, Athorne_Fordy_1987}. In particular, the Manakov system (more precisely, its Hermitian version) is associated with the symmetric space ${\rm SU}(n+1)/{\rm S}({\rm U}(1)\times{\rm U}(n))$ of AIII type in the Helgason classification, and the Kulish--Sklyanin system corresponds to the symmetric space ${\rm SO}(n+2)/({\rm SO}(n)\times{\rm SO}(2))$ of BDI type. We restrict ourselves to these two examples, since the equations for them admit a convenient coordinate-free notation using only vectors and scalar products, but, in principle, the study of similarity reductions makes sense in a broader setting as well. Painlev\'e type reductions for some matrix equations, also related to systems from \cite{Fordy_Kulish_1983, Athorne_Fordy_1987}, were considered in our previous paper \cite{Adler_Sokolov_2021}.

2) Similarity reduction with respect to a one-parameter symmetry subgroup generally leads not directly to a Painlev\'e equation, but to some equation of higher order in derivatives. The order is then reduced by first integrals and/or by passing to invariants of the remaining subgroup of transformations. This is true for both scalar and vector equations, but in the vector case, the reducing of order becomes more complicated and cannot always be achieved while maintaining vector notation, without using the coordinates. Therefore, the name ``vector Painlev\'e equations'' which we use is rather arbitrary and reflects the origin of the systems under consideration rather than their structure.

3) The group of classical symmetries is expanded by linear transformations from the groups ${\rm GL}(n)$ or ${\rm O}(n)$. Thanks to this, an arbitrary constant matrix (of general form or skew-symmetric) appears in the reduced equations, which turns out to be very important, because it eliminates the degeneracy inherent in any isotropic vector ODE. For example, the mixed scalar-vector generalization of the equation \ref{P1} (see section \ref{s:KdV}) has the form
\begin{equation}\label{KdV-Gal'}
 u''=\frac{3}{2}u^2-\frac{3}{2}({\bf v},{\bf v})+z,\quad 
 {\bf v}'''=3u{\bf v}'+3u'{\bf v}+A{\bf v},\quad 
 u\in{\mathbb R},~~ {\bf v}\in{\mathbb R}^n,~~ A+A^T=0.
\end{equation}
If $A=0$ then the order is reduced by integration and we arrive to the system
\[
 u''=\frac{3}{2}u^2-\frac{3}{2}({\bf v},{\bf v})+z,\quad {\bf v}''=3u{\bf v}+{\bf b}.
\]
This system has a number of interesting properties, however, its main drawback is that the number of unknown functions can be reduced from $n+1$ to 3 (see Conclusion), and therefore, unlike the system (\ref{KdV-Gal'}), it is not a ``genuine'' $n$-component analog of the \ref{P1} equation. This does not mean, however, that this system is of no interest for small $n$.

\paragraph{Outline of the article.} We begin the presentation with the NLS as the main model. In Section \ref{s:Manakov}, we describe Painlev\'e type reductions for the Manakov system. They lead to nonautonomous deformations of the Garnier system \cite{Garnier_1919} 
\begin{equation}\label{Garnier}
 {\bf u}''+2({\bf u},{\bf v}){\bf u}+\Omega {\bf u}=0,\quad {\bf v}''+2({\bf u},{\bf v}){\bf v}+\Omega {\bf v}=0
\end{equation}
where $\Omega$ is a diagonal matrix. Recall that this is a Liouville integrable system and its connection with finite-dimensional reductions of integrable equations is well known, see in particular \cite{Choodnovsky_1978, Antonowicz_Rauch-Wojciechowski_1990, Suris_2003}. One of its non-autonomous deformations was studied in \cite{Adler_Kolesnikov_2023, Domrin_Suleimanov_2025} in connection with string equations for the (scalar) KdV hierarchy; related systems were introduced in \cite{Orlov_Rauch-Wojciechowski_1993}. The second deformation is apparently new, but nevertheless, it turns out to be rather unexpectedly related to another important Painlev\'e type system, namely, the quasiperiodic closure of the dressing chain studied by Veselov and Shabat \cite{Veselov_Shabat_1993}. We discuss this connection separately in Section \ref{s:dc}.

Section \ref{s:KS} describes similarity reductions for the Kulish--Sklyanin system. In the autonomous case, a generalization of the Garnier system to the symmetric space of BDI type appears, which was obtained, together with generalizations to other symmetric spaces, in paper \cite{Fordy_Wojciechoski_Marshall_1986}; see also \cite[Chapter 25]{Suris_2003} for a detailed presentation of such results as Hamiltonian structure, Lax representations, Liouville integrability, and B\"acklund transformations. The Painlev\'e type similarity reductions define non-autonomous deformations for this system, which apparently has not been studied before.

In Section \ref{s:mKdV} we consider two versions of the mKdV equation and related third-order systems of Garnier type. Like all other examples, they contain an arbitrary constant matrix $A$ (which is skew-symmetric for this model). Section \ref{s:A0} is devoted to the fully isotropic case $A=0$ for these systems. In this case, the system can be restricted to the invariants of the group ${\rm O}(n)$, which allows the overall order of the system to be reduced to 4, regardless of the dimension of $n$. Of course, intermediate cases of degeneracy are also worth studying, but this is beyond the scope of our article.

Section \ref{s:KdV} is devoted to vector generalizations of the H\'enon--Heiles system
\begin{equation}\label{scalhh}
u''= a u^2+b v^2, \qquad v''= c u v.
\end{equation}
It is easy to see that $\sigma=a/c$ is an invariant of scaling transformations. It is known (see, for example, \cite{ConteMusette}) that the system \eqref{scalhh} is integrable if and only if {\bf 1.} $\sigma =\frac{1}{2}$, {\bf 2.} $\sigma=3$, {\bf 3.} $\sigma= 8$. Non-autonomous deformations of scalar systems \eqref{scalhh} were studied in \cite{Hone_1998}. In \cite{Antonowicz_Rauch-Wojciechowski_1992}, autonomous vector generalizations of the cases {\bf 1}--{\bf 3} were constructed using the eigenfunctions of the Lax operators. Similarity reductions of the vector KdV equation in Section \ref{s:KdV} lead to nonautonomous vector generalizations of case {\bf 1}, one of which is the system \eqref{KdV-Gal'}. Their autonomous versions differ from those studied in \cite{Antonowicz_Rauch-Wojciechowski_1992, Suris_2003}.

For all obtained Painlev\'e type systems, we present the isomonodromic Lax pairs 
\begin{equation}\label{AB-pair}
 {\cal A}'={\cal B}_\zeta+[{\cal B},{\cal A}]
\end{equation}
where $\zeta$ is the spectral parameter, which are inherited under similarity reductions from known zero curvature representations
\begin{equation}\label{UV-pair}
 {\cal U}_t={\cal V}_x+[{\cal V},{\cal U}]
\end{equation}
for respective PDE systems. The procedure for prolongation of a similarity reduction to linear problems is fairly standard and we omit it. The matrices involved have a block structure; when writing them, ${\bf u}$ and ${\bf v}$ are interpreted as column vectors, and row vectors are denoted as ${\bf u}^T$ and ${\bf v}^T$; the identity matrix of size $n\times n$ is denoted as $I_n$. 

\section{Manakov system and deformations of Garnier system}\label{s:Manakov}

The first vector generalization of NLS, the Manakov system \cite{Manakov_1973}, reads
\begin{equation}\label{NLS-1}
 {\bf u}_t={\bf u}_{xx}+2({\bf u},{\bf v}){\bf u},\qquad -{\bf v}_t={\bf v}_{xx}+2({\bf u},{\bf v}){\bf v}.
\end{equation}
In the scalar case, it admits classical symmetries with generators
\[
\begin{aligned}
 & u_{\tau_0}=u,~~ v_{\tau_0}=-v && \text{hyperbolic rotation of $u$ and $v$},\\
 & u_{\tau_1}=u_x,~~ v_{\tau_1}=v_x && \text{$x$-translation},\\
 & u_{\tau_2}=u_t,~~ v_{\tau_2}=v_t && \text{$t$-translation},\\
 & u_{\tau_3}=2tu_x+xu,~~ v_{\tau_3}=2tv_x-xv && \text{Galilean transformation},\\
 & u_{\tau_4}=2tu_t+xu_x+u,~~ v_{\tau_4}=2tv_t+xv_x+v &\qquad& \text{scaling}.
\end{aligned}
\]
These transformations are preserved in the vector case, but the first one is generalized and replaced by invariance with respect to the transformations ${\bf u}\mapsto P{\bf u}$, ${\bf v}\mapsto(P^T)^{-1}{\bf v}$ with $P\in{\rm GL}(n)$, which correspond to infinitesimal generators with arbitrary matrices $A$:
\[
 {\bf u}_\tau=A{\bf u},\qquad {\bf v}_\tau=-A^T{\bf v}.
\]
The solutions of (\ref{NLS-1}) in the form of a traveling wave modulated by matrix exponential are defined by
\begin{equation}\label{uv-wave}
 {\bf u}(x,t)=e^{-kz-tA}{\bf u}(z),\quad {\bf v}(x,t)=e^{kz+tA^T}{\bf v}(z),\quad z=x-2kt
\end{equation}
where ${\bf u}(z)$ and ${\bf v}(z)$ satisfy the system
\[
 {\bf u}''+2({\bf u},{\bf v}){\bf u}+(A-k^2){\bf u}=0,\qquad {\bf v}''+2({\bf u},{\bf v}){\bf v}+(A^T-k^2){\bf v}=0.
\]
This coincides with the classical Garnier system (\ref {Garnier}) if the matrix $A$ is diagonal, which corresponds to the generic case when all eigenvalues are simple. Recall that this system is Hamiltonian and has $2n$ first integrals in involution (see \cite{Suris_2003} and references therein).

Similarity reductions with respect to other one-parameter subgroups are generalized in a similar way. Galilean-invariant solutions are defined by the formulas
\begin{equation}\label{uv-Gal}
 {\bf u}(x,t)=e^{r-tA}{\bf u}(z),\quad 
 {\bf v}(x,t)=e^{-r+tA^T}{\bf v}(z),\quad 
 r=\frac{t^3}{6}-\frac{xt}{2},\quad 
 z=x-\frac{t^2}{2}.   
\end{equation}
As the result of this substitution, (\ref{NLS-1}) turns into the system
\begin{equation}\label{Garnier-1}
 {\bf u}''+2({\bf u},{\bf v}){\bf u}+\frac{1}{2}z{\bf u}+A{\bf u}=0,\qquad
 {\bf v}''+2({\bf u},{\bf v}){\bf v}+\frac{1}{2}z{\bf v}+A^T{\bf v}=0,
\end{equation}
which defines a non-autonomous deformation of the Garnier system. If ${\bf u}$ and ${\bf v}$ are scalars, then, as we show below, this system reduces to the equation \ref{P34} (or \ref{P2}, which is equivalent).

The similarity reduction with respect to the scaling subgroup is defined by the formulas
\begin{equation}\label{uv-self}
 {\bf u}(x,t)=t^{-1/2}e^{-\log(t)A}{\bf u}(z),\quad
 {\bf v}(x,t)=t^{-1/2}e^{\log(t)A^T}{\bf v}(z),\quad 
 z=t^{-1/2}x.
\end{equation} 
As the result of this substitution, (\ref{NLS-1}) turns into the system
\begin{equation}\label{Garnier-2}
 {\bf u}''+2({\bf u},{\bf v}){\bf u}+\frac{1}{2}(z{\bf u})'+A{\bf u}=0,\qquad
 {\bf v}''+2({\bf u},{\bf v}){\bf v}-\frac{1}{2}(z{\bf v})'+A^T{\bf v}=0,
\end{equation}
which defines the second deformation of the Garnier system. In the scalar case, it reduces to the \ref{P4} equation (see Example \ref{ex:P4} below).

The isomonodromic representations (\ref{AB-pair}) for both systems (\ref{Garnier-1}) and (\ref{Garnier-2}) correspond to the same matrix
\[
 {\cal B}=\zeta{\cal B}_0+{\cal B}_1,\quad
 {\cal B}_0=\begin{pmatrix}
  1-\nu & 0\\
   0 & I_n  
 \end{pmatrix},\quad
 {\cal B}_1=\begin{pmatrix}
  0 & -\bf v^T\\
  \bf u & 0  
 \end{pmatrix},
\]
where $\nu\ne0$ is an arbitrary number. The matrix $\cal A$ for the system (\ref{Garnier-1}) is 
\[
 {\cal A}=(2\nu^2\zeta^2+z){\cal B}_0+2\nu^2\zeta{\cal B}_1+2\nu{\cal B}_2,\qquad 
 {\cal B}_2=\begin{pmatrix}
   -({\bf u},{\bf v}) & ({\bf v}')^T\\
    {\bf u}' & {\bf u}{\bf v}^T+A  
 \end{pmatrix},
\]
and for the system (\ref{Garnier-2})
\[
 {\cal A}=(2\nu+z\zeta^{-1})(\zeta{\cal B}_0+{\cal B}_1)+2\zeta^{-1}{\cal B}_2.
\]
Notice that the representations with different parameter $\nu\ne0$ are equivalent to each other, up to a suitable scaling of $\zeta$. This parameter is due to the fact that the matrices from any zero-curvature representation are defined up to adding arbitrary constant scalar matrices. In the $\operatorname{gl}_2$ case (that is, for $n=1$), it is common to fix this arbitrariness by choosing traceless matrices, which correspond to $\nu=2$. For the $\operatorname{gl}_{n+1}$ case this is also possible, by setting $\nu=n+1$, however, this choice may be less convenient than, for instance, $\nu=1$. By this reason we prefer to keep $\nu$ arbitrary. The same remark is applied to the matrices related to the equation (\ref{P21}) and the system (\ref{P21-pq}).

Both systems (\ref{Garnier-1}) and (\ref{Garnier-2}) have an incomplete set of first integrals, which allows for a partial reducing of order. This is more simply arranged for the system (\ref{Garnier-1}). If we turn to the components of the vectors and take $A=\operatorname{diag}(\alpha_1,\dots,\alpha_n)$ (the generic case), then
\[
 u_iv'_i-u'_iv_i=c_i=\operatorname{const},\qquad i=1,\dots,n,
\]
while the products $y_i=u_iv_i$ satisfy the equations
\begin{equation}\label{P34-n}
 y''_i=\frac{(y'_i)^2-c^2_i}{2y_i}-y_i(4y_1+\dots+4y_n+z+2\alpha_i),\qquad i=1,\dots,n.
\end{equation} 
For $n=1$, it is easy to see that the equation for $2y_1$ coincides with equation \ref{P34} up to adding a constant to $z$. The system (\ref{P34-n}) can be viewed as a multicomponent generalization of this Painlev\'e equation. This system was studied in \cite{Adler_Kolesnikov_2023} within a construction related to negative KdV symmetries, see also papers \cite{Antonowicz_Rauch-Wojciechowski_1990, Orlov_Rauch-Wojciechowski_1993}. Its analytic properties were investigated in the paper \cite{Domrin_Suleimanov_2025}.

Reducing of order for system (\ref{Garnier-2}) is less obvious (but also more interesting); it is discussed in the next section.

\section{System (\ref{Garnier-2}) and the quasiperiodic dressing chain}\label{s:dc}

This section is devoted to the following system of ODEs on the variables $\psi_i(x)$ and $\varphi_i(x)$:
\begin{equation}\label{G2}
 \psi''_i=(q-\alpha_i-1)\psi_i,\quad
 \varphi''_i=(q-\alpha_i+1)\varphi_i,\quad
 q=x^2-2(\psi_1\varphi_1+\dots+\psi_n\varphi_n).
\end{equation}
Here, the independent variable $x$ does not match the one from the original NLS system (\ref{NLS-1}), which should not lead to misunderstanding. Equations (\ref{G2}) are related to the second deformation of the Garnier system (\ref{Garnier-2}) by the changes of variables
\[
 z=2x,\quad 2u(z)=e^{-x^2/2}\psi(x),\quad 2v(z)=e^{x^2/2}\varphi(x),\quad 
 4A=\operatorname{diag}(\alpha_1,\dots,\alpha_n),
\]
which eliminate terms with first derivatives and simplify the coefficients. The parameters $\alpha_i$ are assumed to be distinct. We also assume that $\psi_i$ and $\varphi_i$ do not vanish identically. This does not lead to a loss of generality, since if, for example, $\psi_i=0$, then the system reduces to the same system of lower dimension and an additional equation for $\varphi_i$.
 
The overall order of the system is reduced from $4n$ to $3n$ by passing to variables
\[
 y_i=\psi_i\varphi_i,\qquad w_i=\psi_i\varphi'_i-\psi'_i\varphi_i,
\]
which leads to equations
\begin{equation}\label{yG2}
 2y_iy''_i=(y'_i)^2+4(q-\alpha_i)y^2_i-w^2_i,\quad w'_i=2y_i,\quad q=x^2-2y_1-\dots-2y_n.
\end{equation}
The variables $\psi_i$ and $\varphi_i$ are reconstructed from $y_i$ and $w_i$ by a quadrature:
\[
 \psi_i= y^{1/2}_ie^{-\chi_i},\quad \varphi_i= y^{1/2}_ie^{\chi_i},\quad \chi_i=\int\frac{w_i\,dx}{2y_i},
\]
which allows us to consider the systems (\ref{G2}) and (\ref{yG2}) as equivalent. Next, reducing of order to $2n$ is carried out as follows. Let us introduce the polynomials
\begin{gather*}
 a(\lambda)=(\lambda-\alpha_1)\cdots(\lambda-\alpha_n)
  =\lambda^n-\lambda^{n-1}\sum^n_{i=1}\alpha_i+\dotsc,\\
 y(x,\lambda)=a(\lambda)\left(1-\frac{y_1}{\lambda-\alpha_1}
     -\dots-\frac{y_n}{\lambda-\alpha_n}\right)
  =\lambda^n-\lambda^{n-1}\sum^n_{i=1}(y_i+\alpha_i)+\dotsc,\\ 
 w(x,\lambda)=a(\lambda)\left(2x-\frac{w_1}{\lambda-\alpha_1}
     -\dots-\frac{w_n}{\lambda-\alpha_n}\right)
  =2x\lambda^n+\dotso. 
\end{gather*}

\begin{theorem}
The polynomials $y$ and $w$ satisfy the equations
\begin{equation}\label{ywc}
 2yy''=(y')^2+4(q-\lambda)y^2-w^2+c,\quad w'=2y,\quad c=4a(\lambda)b(\lambda)
\end{equation}
where $b(\lambda)=(\lambda-\beta_1)\cdots(\lambda-\beta_{n+1})$ is a polynomial with constant coefficients, such that
\begin{equation}\label{ab}
 \sum^n_{i=1}\alpha_i=\sum^{n+1}_{i=1}\beta_i.
\end{equation}
\end{theorem}

\begin{proof}
The equation $w'=2y$ follows immediately from the equations $w'_i=2y_i$. Differentiation of the first equation (\ref{yG2}) gives
\[
 y'''_i=4(q-\alpha_i)y'_i+2q'y_i-2w_i.
\]
Let us multiply these equalities by $-a(\lambda)/(\lambda-\alpha_i)$ and sum over $i$. This gives
\begin{align*}
 y''' &= 4(q-\lambda)y' -4a\sum^n_{i=1}y'_i +2q'(y-a)-2w+4ax\\
  &= 4(q-\lambda)y'+2q'y-2w - 2a\left(q'+2\sum^n_{i=1}y'_i-2x\right)\\
  &= 4(q-\lambda)y'+2q'y-2w,
\end{align*}
where we have used the last equation (\ref{yG2}). By multiplying the obtained relation by $2y$ and integrating, we arrive at the first equation (\ref{ywc}). The integration constant $c(\lambda)$ is a polynomial, since $y$ and $w$ are polynomials, and from the comparison of the highest degrees it follows that $c=4\lambda^{2n+1}+\dotso$. From the definition of $y$ and $w$ it is easy to see that
\[
 y(x,\alpha_i)=-a'(\alpha_i)y_i,\qquad w(x,\alpha_i)=-a'(\alpha_i)w_i,
\]
moreover $a'(\alpha_i)\ne0$, since the zeroes $\alpha_i$ are simple by assumption. Then a comparison of Eqs. (\ref{yG2}) and (\ref{ywc}) proves that $c(\alpha_i)=0$, that is, the polynomial $c$ is divisible by $a$. Let $c=4ab$ and $\beta_i$ be the zeros of $b$. Then the coefficients of $\lambda^{2n}$ in (\ref{ywc}) yield the equality
\[
 q-x^2+2\sum^n_{i=1}y_i+\sum^n_{i=1}\alpha_i-\sum^{n+1}_{i=1}\beta_i=0
\]
and the relation (\ref{ab}) follows by taking into account the definition of $q$ from (\ref{yG2}).
\end{proof}

Thus, the first equation (\ref{ywc}) gives $n$ first integrals for the system (\ref{yG2}), the values of which are determined by the constants $\beta_i$.

\begin{example}\label{ex:P4} 
For the simplest case $n=1$, the system (\ref{yG2}) reads
\begin{equation}\label{yG2-1}
 2y_1y''_1=(y'_1)^2+4(x^2-2y_1-\alpha_1)y^2_1-w^2_1,\qquad w'_1=2y_1.
\end{equation}
Let us substitute the polynomials $a=\lambda-\alpha_1$, $y=\lambda-y_1-\alpha_1$, $w=2x(\lambda-\alpha_1)-w_1$ and $c=4(\lambda-\alpha_1)(\lambda-\beta_1)(\lambda-\beta_2)$ into (\ref{ywc}). By equating the coefficients of $\lambda$, we obtain Eqs.~(\ref{yG2-1}) themselves, the relation $\alpha_1=\beta_1+\beta_2$, and an additional equation
\[
 \frac{(y'_1)^2-w^2_1}{y_1}+4y^2_1-4(x^2-\alpha_1)y_1+4xw_1+4\beta_1(\alpha_1-\beta_1)=0,
\]
which serves as a first integral for (\ref{yG2-1}), as can be directly verified. Thus, the order of the system is reduced to 2. In principle, it is possible to transform this system to a single second-order equation for $y_1$ or $w_1$, but it will be quadratic in the second derivative. It turns out that the ``more suitable'' variable is
\[
 f=\frac{y'_1+w_1}{2y_1}-x,
\]
which, as one can verify, satisfies the equation \ref{P4} with the parameter values $\alpha=1-\alpha_1/2$ and $\beta=-(2\beta_1-\alpha_1)^2/2$. The origin of this substitution is explained in Proposition \ref{pr:G2ff} below.
\end{example} 

For the autonomous Garnier system (\ref {Garnier}), the equations for $y$ and $w$ differ in that the potential $q$ does not contain the term $x^2$ and the polynomial $w$ is constant. In this case, the equations (\ref {ywc}) reduce to the system of Dubrovin equations of order $n$ for the zeroes of the polynomial $y$, defining an $n$-gap potential $q$, where the endpoints of the gaps are determined by the zeroes of the polynomial $w^2-c$ \cite{Novikov_1974, Dubrovin_1975}. The equivalence between the Novikov--Dubrovin equations and the Garnier system was established in the article \cite{Choodnovsky_1978} (see also \cite[Chapter 22]{Suris_2003}).

The passage to the zeroes of $y$ is possible for our system as well, but this leads to a significantly more complicated system of order $2n$, since the polynomial $w^2-c$ also evolves according to the equation $w'=2y$. It turns out that a simpler form of the system is related to the quasiperiodic closure of the dressing chain studied in the celebrated paper by Veselov and Shabat \cite{Veselov_Shabat_1993}. Recall that a family of Schr\"odinger operators $L=-d^2/dx^2+q(x)$ was introduced in this paper, such that the action of elementary Darboux transformations $q_i\mapsto q_{i+1}$ defined by the formulas
\begin{equation}\label{fq}
 q_i=f'_i+f^2_i+\gamma_i,\qquad q_{i+1}=-f'_i+f^2_i+\gamma_i
\end{equation}
brings after $N$ steps to the original potential shifted by a constant amount: $q_{i+N}=q_i+2\delta$. This quasi-periodicity condition leads to a closed system of ODEs for the variables $f_i$:
\begin{equation}\label{ff}
 f'_i+f'_{i+1}=f^2_i-f^2_{i+1}+\gamma_i-\gamma_{i+1},\quad 
 f_{i+N}=f_i,\quad \gamma_{i+N}=\gamma_i+2\delta.
\end{equation}
The properties of this system depend on whether the parameter $\delta$ is zero or not, and on the parity of $N$. In the strictly periodic case $\delta=0$, as shown in \cite{Veselov_Shabat_1993}, the system (\ref{ff}) for odd $N=2n+1$ is equivalent to the Novikov--Dubrovin equations for $n$-gap potentials. Thus, the equivalence with the Garnier system (\ref{Garnier}) is also established. The case of even $N=2n+2$ also reduces to $n$-gap potentials, but in a somewhat more complex way.

If $\delta\ne0$ then the system (\ref{ff}) admits only one non-autonomous first integral $f_1+\dots+f_N+\delta x=\operatorname{const}$, but it passes the Kowalevskaya--Painlev\'e test. For $N=1$ one obtains the usual harmonic oscillator $q_1(x)=\delta^2x^2$, for $N=2$ one obtains the oscillator on a half-line, and $N>2$ leads to transcendental generalizations of these potentials with a spectrum consisting of $N$ arithmetic progressions with step $2\delta$; in particular, the case $N=3$ is associated with the \ref{P4} equation, and $N=4$ with the Painlev\'e P$_5$ equation \cite{Veselov_Shabat_1993, Adler_1993}. 

\begin{proposition}\label{pr:G2ff}
The deformation of the Garnier system (\ref{G2}), in the form (\ref{yG2}), is equivalent to the system (\ref{ff}) with $N=2n+1$ and $\delta=1$.
\end{proposition}

\begin{proof}
We will use the matrix representation of the Darboux transformation of general form and the theorem on its decomposition into elementary ones (see \cite{Veselov_Shabat_1993}). Potentials $q$ and $\hat q$ are said to be related by a Darboux transformation if there exists a matrix $F$ that is polynomial in $\lambda$ and satisfies the equation
\begin{equation}\label{FQ}
 F'=\widehat QF-FQ,\qquad Q=\begin{pmatrix} 0 & 1 \\ q-\lambda & 0 \end{pmatrix}.
\end{equation}
By writing this equation elementwise, it is easy to show that such a matrix must be of the form
\[
 F=\begin{pmatrix} 
    \dfrac{1}{2}(w-y') & y \\ 
    \dfrac{1}{4y}(w^2-(y')^2-c) & \dfrac{1}{2}(w+y') 
   \end{pmatrix}
\]
where $y$, $w$ and $c=4\det F$ are polynomials in $\lambda$ satisfying the equations
\begin{equation}\label{ywc'}
 2yy''=(y')^2+2(q+\hat q-2\lambda)y^2-w^2+c,\quad w'=(\hat q-q)y,\quad c'=0.
\end{equation}
The first equation implies that the matrix element $F_{21}$ is equal to $\frac{1}{2}(q+\hat q-2\lambda)y-\frac{1}{2}y''$ and is also a polynomial in $\lambda$. The elementary Darboux transformation corresponds to the choice $y=1$, $w=-2f_i(x)$ and $c=4(\lambda-\gamma_i)$, that is, to the matrix
\[
 F_i=\begin{pmatrix} 
    -f_i & 1 \\ 
    f^2_i+\gamma_i-\lambda & -f_i 
   \end{pmatrix},
\]
and moreover, Eqs. (\ref{ywc'}) with $q=q_i$ and $\hat q=q_{i+1}$ are equivalent to the transform (\ref{fq}). A theorem holds that the general Darboux transformation can be represented by a sequence of elementary ones, and any polynomial matrix $F$ satisfying (\ref{FQ}) has the form $F=F_N\cdots F_1$, where $\gamma_i$ are the zeros of $c(\lambda)=4\det F$, taking multiplicities into account. Such a factorization is unique up to the numbering of the zeros. By induction, it is easy to show that
\[
\begin{array}{lll}
 w=-2(f_1+\dots+f_N)\lambda^n+\dotsc,& y=\lambda^n+\dots &\text{ for } N=2n+1,\\
 w=2\lambda^{n+1}+\dotsc, & y= -(f_1+\dots+f_N)\lambda^n+\dots &\text{ for } N=2n+2.
\end{array}
\]
We only need odd $N=2n+1$. Consider the quasiperiodic closure $q_{N+1}=q_1+2\delta$, then, given the first integral $w=2\delta(x-x_0)\lambda^n+\dotsc$ of the system (\ref{ff}), we can set $x_0=0$ without loss of generality. Equations (\ref{ywc'}) with $q=q_1$ and $\hat q=q_{N+1}$ take the form
\[
 2yy''=(y')^2+4(q_1+\delta-\lambda)y^2-w^2+c,\quad 
 w'=2\delta y,\quad 
 c=4(\lambda-\gamma_1)\dots(\lambda-\gamma_{2n+1}).
\]
This coincides with (\ref{ywc}) for $\delta=1$, $q=q_1+\delta$ and $c=4ab$, as required.
\end{proof}

Thus, we have established an implicit but rather interesting correspondence between the system (\ref{yG2}) and the quasiperiodic dressing chain with odd period $2n+1$. A noteworthy circumstance is that the parameters $\alpha_1,\dots,\alpha_n$ and the first integrals $\beta_1,\dots,\beta_{n+1}$ of the system (\ref{yG2}) turn out to be on equal footing in the chain, since both are zeros of the polynomial $c(\lambda)$. Note that they enter the dressing chain only through differences, which reflects the possibility of simultaneously adding an arbitrary constant to $q$ and  $\lambda$ in the Schr\"odinger equation. However, in the Garnier system (\ref{G2}), the potential $q$ is fixed by the formula $q=x^2-2\sum\psi_i\varphi_i$, which fixes $\lambda$ as well and leads to the normalization (\ref{ab}). This is not a limitation, since in the reverse transition from the chain to the Garnier system, we first arbitrarily divide the set of parameters $\gamma_i$ into subsets of $\alpha_i$ and $\beta_i$, and then the condition (\ref{ab}) is provided by adding a suitable constant to all parameters.
 
The question on existence of Garnier type systems corresponding to the dressing chain with even period remains open. As already noted, the cases $N=3$ and $N=4$ correspond to different Painlev\'e transcendents, in contrast to the strictly periodic closure, for which both systems are solved in terms of elliptic functions. Apparently, this means that the even case is not related to the system (\ref{G2}). Perhaps it corresponds to some other non-autonomous generalization, possibly with a term like $x^{-2}$ added to the potential $q$, as the example with $N=2$ suggests.

\section{Kulish--Sklyanin system}\label{s:KS}

The system
\begin{equation}\label{NLS-2}
 {\bf u}_t={\bf u}_{xx}+2({\bf u},{\bf v}){\bf u}-({\bf u},{\bf u}){\bf v},\qquad 
-{\bf v}_t={\bf v}_{xx}+2({\bf u},{\bf v}){\bf v}-({\bf v},{\bf v}){\bf u}
\end{equation}
introduced by Kulish and Sklyanin \cite{Kulish_Sklyanin_1981} is an example of multicomponent generalizations of the NLS equations associated with symmetric spaces \cite{Fordy_Kulish_1983}. Unlike the Manakov system (\ref{NLS-1}), which is invariant under ${\rm GL}(n)$, this system is invariant only under orthogonal transformations $({\bf u},{\bf v})\mapsto(P{\bf u},P{\bf v})$, where $PP^t=I$, which corresponds to generators ${\bf u}_\tau=A{\bf u}$, ${\bf v}_\tau=A{\bf v}$ with skew-symmetric matrices $A$. Otherwise, the classical symmetry group remains the same. In this section, we assume that
\[
 A=-A^t,
\] 
then the formulas (\ref{uv-wave}), (\ref{uv-Gal}) and (\ref{uv-self}) work for reductions of the system (\ref{NLS-2}) as well. Its traveling wave solutions are governed by the autonomous system
\begin{equation}\label{KSGarnier}
 {\bf u}''+2({\bf u},{\bf v}){\bf u}-({\bf u},{\bf u}){\bf v}+(A-k^2){\bf u}=0,\quad 
 {\bf v}''+2({\bf u},{\bf v}){\bf v}-({\bf v},{\bf v}){\bf u}-(A+k^2){\bf v}=0,
\end{equation}
which defines a generalization of the Garnier system to a symmetric space of BDI type \cite{Fordy_Wojciechoski_Marshall_1986}. This is a Liouville integrable system (see \cite[Chapter 25.5]{Suris_2003} for details, in slightly different notation).

The Galilean similarity reduction (\ref{uv-Gal}) leads to the first deformation of the system (\ref{KSGarnier})
\begin{equation}\label{KSGarnier-1}
 {\bf u}''+2({\bf u},{\bf v}){\bf u}-({\bf u},{\bf u}){\bf v}+\frac{1}{2}z{\bf u}+A{\bf u}=0,\quad
 {\bf v}''+2({\bf u},{\bf v}){\bf v}-({\bf v},{\bf v}){\bf u}+\frac{1}{2}z{\bf v}-A{\bf v}=0,
\end{equation}
and the scaling similarity reduction (\ref{uv-self}) leads to the second deformation 
\begin{equation}\label{KSGarnier-2}
 {\bf u}''+2({\bf u},{\bf v}){\bf u}-({\bf u},{\bf u}){\bf v}+\frac{1}{2}(z{\bf u})'+A{\bf u}=0,\quad
 {\bf v}''+2({\bf u},{\bf v}){\bf v}-({\bf v},{\bf v}){\bf u}-\frac{1}{2}(z{\bf v})'-A{\bf v}=0.
\end{equation}
Like systems (\ref{Garnier-1}) and (\ref{Garnier-2}), these systems define vector generalizations of equations \ref{P34} and \ref{P4}, respectively. The autonomous system (\ref{KSGarnier}) admits the standard isospectral Lax representation ${\cal A}'=[{\cal B},{\cal A}]$ with
\begin{gather*}
 {\cal B}=\zeta{\cal B}_0 +{\cal B}_1,\qquad
 {\cal B}_0=\begin{pmatrix}
  -1 & 0 & 0\\
   0 & 0 & 0\\
   0 & 0 & 1
 \end{pmatrix},\qquad
 {\cal B}_1=\begin{pmatrix}
  0 & -{\bf v}^T & 0\\
  {\bf u} &  0   & -{\bf v}\\
  0 &  {\bf u}^T & 0
 \end{pmatrix},\\
 {\cal A}=2(\zeta^2-k^2){\cal B}_0+2\zeta{\cal B}_1+2{\cal B}_2,\qquad
 {\cal B}_2=\begin{pmatrix}
   -({\bf u},{\bf v}) & ({\bf v}')^T & 0\\
   {\bf u}' & {\bf u}{\bf v}^T-{\bf v}{\bf u}^T+A & {\bf v}'\\
   0 & ({\bf u}')^T & ({\bf u},{\bf v}) 
 \end{pmatrix}.
\end{gather*}
The isomonodromic Lax representations (\ref{AB-pair}) for both systems (\ref{KSGarnier-1}) and (\ref{KSGarnier-2}) correspond to the same matrix ${\cal B}=\zeta{\cal B}_0 +{\cal B}_1$; the matrix $\cal A$ for the system (\ref{KSGarnier-1}) is
\[
 {\cal A}=(2\zeta^2+z){\cal B}_0+2\zeta{\cal B}_1+2{\cal B}_2,
\]
and for the system (\ref{KSGarnier-2}) it is
\[
 {\cal A}=(2+z\zeta^{-1})(\zeta{\cal B}_0 +{\cal B}_1)+2\zeta^{-1}{\cal B}_2.
\]

\section{Vector mKdV equations and their isomonodromic reductions}\label{s:mKdV}

The Manakov system admits the third-order symmetry
\begin{equation}\label{NLS-1-t3}
\begin{aligned}
 {\bf u}_T&= {\bf u}_{xxx}+3({\bf u},{\bf v}){\bf u}_x+3({\bf u}_x,{\bf v}){\bf u},\\
 {\bf v}_T&= {\bf v}_{xxx}+3({\bf u},{\bf v}){\bf v}_x+3({\bf u},{\bf v}_x){\bf v},
\end{aligned}
\end{equation}
and the Kulish--Sklyanin system admits the symmetry
\begin{equation}\label{NLS-2-t3}
\begin{aligned}
 {\bf u}_T&={\bf u}_{xxx}+3({\bf u},{\bf v}){\bf u}_x+3({\bf u}_x,{\bf v}){\bf u}-3({\bf u},{\bf u}_x){\bf v},\\
 {\bf v}_T&={\bf v}_{xxx}+3({\bf u},{\bf v}){\bf v}_x+3({\bf u},{\bf v}_x){\bf v}-3({\bf v},{\bf v}_x){\bf u}.
\end{aligned}
\end{equation}
In both cases, the reduction ${\bf v}=-{\bf u}$ is possible, which leads, respectively, to two vector generalizations of the mKdV equation \cite{Athorne_Fordy_1987, Sokolov_Svinolupov_1994}: 
\begin{equation}\label{mKdV-1}
 {\bf u}_t={\bf u}_{xxx}-3({\bf u},{\bf u}){\bf u}_x-3({\bf u},{\bf u}_x){\bf u}
\end{equation}
and
\begin{equation}\label{mKdV-2}
 {\bf u}_t={\bf u}_{xxx}-3({\bf u},{\bf u}){\bf u}_x.
\end{equation}
Both equations are invariant with respect to the one-parameter group of dilatations ${\bf u}(x,t,a)$ $=e^a{\bf u}(e^ax,e^{3a}t)$, as well as to orthogonal transformations of ${\bf u}$. The corresponding self-similar solutions are of the form
\begin{equation}\label{self}
 {\bf u}(x,t)= \varepsilon\tau e^{\log(\tau)A}{\bf u}(z),\quad 
 \tau=t^{-1/3},\quad z=\varepsilon\tau x,\quad 3\varepsilon^3=-1
\end{equation}
where $A$ is an arbitrary skew-symmetric matrix (the multiplier $\varepsilon$ is introduced only for convenience of comparison with the canonical form of the equation \ref{P2}). As the result of this substitution, Eq. (\ref{mKdV-1}) turns into
\begin{equation}\label{P21}
 {\bf u}'''= 3({\bf u},{\bf u}){\bf u}'+3({\bf u},{\bf u}'){\bf u}+z{\bf u}'+{\bf u}+A{\bf u},\quad  A=-A^t,  
\end{equation}
and Eq. (\ref{mKdV-2}) turns into
\begin{equation}\label{P22}
 {\bf u}'''= 3({\bf u},{\bf u}){\bf u}'+z{\bf u}'+{\bf u}+A{\bf u},\quad A=-A^t.
\end{equation}
If ${\bf u}$ is a scalar, then $A=0$, and the equations can be integrated once, resulting in equation \ref{P2} with the parameter $\alpha$ as the integration constant (for (\ref{P22}), an additional scaling is needed which is not essential). For the vector equations, such a reduction of order does not occur, even if we set $A=0$. 

The zero curvature representation (\ref{UV-pair}) for (\ref{mKdV-1}) is given by block matrices of size $(1,n)\times(1,n)$ with an arbitrary parameter $\nu\ne0$:
\begin{gather*}
 {\cal U}=\begin{pmatrix}
   (1-\nu)\lambda & {\bf u}^T \\ 
   {\bf u} & \lambda I_n   
 \end{pmatrix},\\[3mm]
 {\cal V}=\nu^2\lambda^2{\cal U}+{\cal U}_{xx}
 +\nu\lambda\begin{pmatrix}
   ({\bf u},{\bf u}) & -{\bf u}^T_x \\ 
   {\bf u}_x & -{\bf u}{\bf u}^T   
 \end{pmatrix}
 -\begin{pmatrix}
   0 & 2({\bf u},{\bf u}){\bf u}^T \\ 
   2({\bf u},{\bf u}){\bf u} & {\bf u}_x{\bf u}^T-{\bf u}{\bf u}^T_x   
 \end{pmatrix}.
\end{gather*}
For (\ref{mKdV-2}) the matrices are of block size $(1,n,1)\times(1,n,1)$:
\begin{gather*}
 {\cal U}=\begin{pmatrix}
   -\lambda & {\bf u}^T & 0\\ 
    {\bf u} & 0 & {\bf u}\\   
    0 & {\bf u}^T & \lambda 
 \end{pmatrix},\\[3mm]
 {\cal V}=\lambda^2{\cal U}+{\cal U}_{xx}
 +\lambda\begin{pmatrix}
   ({\bf u},{\bf u}) & -{\bf u}^T_x & 0\\ 
   {\bf u}_x & 0 & -{\bf u}_x\\
   0 & {\bf u}^T_x & -({\bf u},{\bf u})\\
 \end{pmatrix}
 -\begin{pmatrix}
   0 & ({\bf u},{\bf u}){\bf u}^T & 0\\ 
   ({\bf u},{\bf u}){\bf u} & 2{\bf u}_x{\bf u}^T-2{\bf u}{\bf u}^T_x & ({\bf u},{\bf u}){\bf u}\\   
   0 & ({\bf u},{\bf u}){\bf u}^T & 0
 \end{pmatrix}.
\end{gather*}
The isomonodromic Lax pairs (\ref{AB-pair}) for (\ref{P21}) and (\ref{P22}) are derived from these zero curvature representations by applying the additional substitution $\lambda=\varepsilon\tau\zeta$, which separates the dependence on $\tau$. For the system (\ref{P21}), this procedure yields the matrices
\begin{equation}\label{P2AB}
 {\cal A}=-\nu^2\zeta(\zeta{\cal B}_0+{\cal B}_1)+z{\cal B}_0-\nu{\cal B}_2 
 -\zeta^{-1}({\cal B}''_1-z{\cal B}_1-{\cal B}_3-{\cal A}_0),\qquad
 {\cal B}=\zeta{\cal B}_0+{\cal B}_1
\end{equation}
where
\begin{equation}\label{P12AB}
\begin{gathered}
 {\cal B}_0=\begin{pmatrix} 1-\nu & 0\\ 0 & I_n \end{pmatrix},\qquad
 {\cal B}_1=\begin{pmatrix} 0 & {\bf u}^T\\ {\bf u} & 0 \end{pmatrix},\qquad
 {\cal B}_2=\begin{pmatrix} ({\bf u},{\bf u}) & -({\bf u}')^T\\ {\bf u}' & -{\bf u}{\bf u}^T \end{pmatrix},\\[3mm]
 {\cal B}_3=\begin{pmatrix} 0 & 2({\bf u},{\bf u}){\bf u}^T\\ 2({\bf u},{\bf u}){\bf u} & {\bf u}'{\bf u}^T-{\bf u}({\bf u}')^T \end{pmatrix},\qquad
 {\cal A}_0=\begin{pmatrix} 0 & 0\\ 0 & A \end{pmatrix}.
\end{gathered}
\end{equation}
Note that if we set $\nu=0$, then in the scalar case the matrix ${\cal A}$ becomes trivial, since the coefficient at $\zeta^{-1}$ turns into a constant matrix due to the first integral (that is, by virtue of the equation \ref{P2}). In the vector case, there is no such a first integral, so the representation formally remains correct for $\nu=0$, but it turns into the degenerate Lax representation for (\ref{P21}) without spectral parameter:
\[
 \tilde{\cal A}'=[\tilde{\cal B},\tilde{\cal A}],\qquad 
 \tilde{\cal A}=\operatorname*{res}_{\zeta=0}{\cal A}
  =-{\cal B}''_1+z{\cal B}_1+{\cal B}_3+{\cal A}_0,\qquad
 \tilde{\cal B}=\tilde{\cal B}_1. 
\]
Characteristic polynomial $\det(\tilde{\cal A}-\mu I_{n+1})$ gives some first integrals for Eq. (\ref{P21}).

For Eq. (\ref{P22}), the matrices ${\cal A}$ and ${\cal B}$ are given by the same formula (\ref{P2AB}) with $\nu=1$ and
\begin{equation}\label{P22AB}
\begin{gathered}
 {\cal B}_0=\begin{pmatrix} -1 & 0 & 0\\ 0 & 0 & 0\\ 0 & 0 & 1 \end{pmatrix},\qquad
 {\cal B}_1=\begin{pmatrix} 0 & {\bf u}^T & 0\\ {\bf u} & 0 & {\bf u}\\ 0 & {\bf u}^T & 0 \end{pmatrix},\qquad
 {\cal B}_2=\begin{pmatrix} 
  ({\bf u},{\bf u}) & -({\bf u}')^T & 0\\
  {\bf u}' & 0 & -{\bf u}'\\ 
  0 & ({\bf u}')^T & -({\bf u},{\bf u})\end{pmatrix},\\[3mm]
 {\cal B}_3=\begin{pmatrix} 
  0 & ({\bf u},{\bf u}){\bf u}^T & 0\\ 
  ({\bf u},{\bf u}){\bf u} & 2{\bf u}'{\bf u}^T-2{\bf u}({\bf u}')^T & ({\bf u},{\bf u}){\bf u}\\
  0 & ({\bf u},{\bf u}){\bf u}^T & 0 
  \end{pmatrix},\qquad
 {\cal A}_0=\begin{pmatrix} 0&0&0\\ 0&A&0\\ 0&0&0 \end{pmatrix}.
\end{gathered}
\end{equation}

\section{Reducing of order in the fully isotropic case}\label{s:A0}

Using the examples from the previous section, we demonstrate that in the completely degenerate case $A=0$, the systems can be restricted to invariants of the group ${\rm O}(n)$, which leads to equations of a fixed order regardless of the dimension $n$. For vector equations of the third order, it is sufficient to consider the variables
\[
 p=({\bf u},{\bf u}),\qquad q=({\bf u}',{\bf u}').
\]
It is easy to verify that, by virtue of Eq. (\ref{P21}) with $A=0$, $p$ and $q$ satisfy the system
\begin{equation}\label{P21-pq}
 p'''= 6pp'+3q'+zp'+2p,\qquad
 q'''= 12pq'+6p'p''+3p''+4zq'-2q.
\end{equation}
Similarly, equation (\ref{P22}) leads to the system
\begin{equation}\label{P22-pq}
 p'''= 3pp'+3q'+zp'+2p,\qquad
 q'''= 12pq'+6p'q+3p''+4zq'-2q.
\end{equation}

\begin{remark}
The system (\ref{P21-pq}) is related to a self-similar reduction of one of the Drinfeld--Sokolov systems \cite{Drinfeld_Sokolov_1985}
\begin{equation}\label{DS-III}
 f_t=f_{xxx}-6ff_x+12g_x,\qquad g_t=-2g_{xxx}+6fg_x.
\end{equation}
Namely, it is not difficult to verify that the latter system admits solutions of the form
\[
 f(x,t)= \varepsilon^2\tau^2f(z),\qquad g(x,t)= \varepsilon^4\tau^4g(z),
\]
where $\tau$, $z$, and $\varepsilon$ are defined as in (\ref{self}) and $f(z)$, $g(z)$ satisfy the equations
\[
 f'''=6ff'-12g'+zf'+2f,\qquad g'''=3fg'-\frac{1}{2}zg'-2g.
\]
This system is connected with (\ref{P21-pq}) by the invertible transformation
\[
 2f=4p+z,\qquad 2g=-q+p^2+zp+z^2/4. 
\]
For Eqs. (\ref{P22-pq}), we were unable to find a similar interpretation as a reduction of any two-component integrable PDE system.
\end{remark}

Below we demonstrate that each of the systems (\ref{P21-pq}) and (\ref{P22-pq}) passes the Kowalevskaya--Painlev\'e test, admits an isomonodromic Lax representation, and has a pair of polynomial first integrals. 

\begin{proposition}
Both systems (\ref{P21-pq}) and (\ref{P22-pq}) possess formal Laurent solutions of the form 
$$
p=\sum_{i=-2}^{\infty} \alpha_i (z-a)^i, \qquad q=\sum_{i=-4}^{\infty} \beta_i (z-a)^i.
$$
with six arbitrary parameters.
\end{proposition}

\begin{proof}
Straightforward calculation proves that for system (\ref{P21-pq}) there are two branches $\alpha_{-2}=1$, $\beta_{-4}=1$ and $\alpha_{-2}=5$, $\beta_{-4}=-15$. For the first branch, the parameters $a$, $\alpha_1$, $\alpha_2$, $\alpha_4$, $\alpha_6$, and $\beta_0$ are free and form a maximal set. There are only four free parameters in the second branch: $a$, $\alpha_4$, $\alpha_6$, and $\beta_0$.    

For (\ref{P22-pq}), the main branch is $\alpha_{-2}=2$, $\beta_{-4}=2$ and it also contains six free parameters $a$, $\alpha_0$, $\alpha_1$, $\alpha_2$, $\alpha_4$, and $\alpha_8$.  
\end{proof}

\begin{proposition} 
i) System (\ref{P21-pq}) admits the isomonodromic Lax representation (\ref{AB-pair}), where
\begin{equation}\label{AB-P21-pq}
 {\cal A}= -\nu^2\zeta(\zeta{\cal A}_2-{\cal A}_1)
  +\nu{\cal A}_0+zI_4+\frac{1}{\zeta}{\cal A}_{-1},\qquad
 {\cal B}=\begin{pmatrix}
  (1-\nu)\zeta & 1 & 0 & 0\\
   p & \zeta & 1 & 0\\
   \frac{1}{2}p' & 0 & \zeta & 1 \\
   \frac{1}{2}p''-q & \frac{3}{2}p'+1 & 3p+z &\zeta
\end{pmatrix},
\end{equation}
\begin{gather*}
 {\cal A}_2=\begin{pmatrix}
  1-\nu & 0 & 0 & 0\\
   0 & 1 & 0 & 0\\
   0 & 0 & 1 & 0 \\
   0 & 0 & 0 & 1
  \end{pmatrix},\quad
 {\cal A}_1=\begin{pmatrix}
   0 & -1 & 0 & 0\\
   -p & 0 & 0 & 0\\
   -\frac{1}{2}p' & 0 & 0 & 0 \\
   q-\frac{1}{2}p'' & 0 & 0 & 0
  \end{pmatrix},\quad
 {\cal A}_0=\begin{pmatrix}
   -p-z & 0 & 1 & 0\\
   -\frac{1}{2}p' & p & 0 & 0\\
   -q & \frac{1}{2}p' & 0 & 0 \\
   -\frac{1}{2}q' & \frac{1}{2}p''-q & 0 & 0
  \end{pmatrix},\\[3mm]
 {\cal A}_{-1}=\begin{pmatrix}
   0 & 2p+z & 0 & -1 \\
 -\frac{1}{2}p''+2p^2+zp+q & \frac{1}{2}p' & -p & 0 \\
 \frac{1}{2}zp'+pp'-\frac{1}{2}q' & q & -\frac{1}{2}p' & 0 \\
 \frac{1}{2}zp''+pp''+\frac{3}{4}(p')^2+\frac{1}{2}p'+pq-\frac{1}{2}q'' & \frac{1}{2}q' & q-\frac{1}{2}p'' & 0 \\
  \end{pmatrix}.
\end{gather*}
ii) System (\ref{P21-pq}) has the following first integrals:
\begin{gather*}
H_1=q'' +2zq-2(2p+z)p''-p'^2-p'+2p(2p+z)^2,\\
H_2= 4q(p''-2q)^2-4q'(p'p''-2p'q-pq')+(p'^2-4pq)(2q''-3p'^2-2p'-12pq-4zq).
\end{gather*}
\end{proposition}

\begin{proof}
Let us demonstrate how the matrices (\ref{AB-P21-pq}) are derived from the matrices ${\cal A}$ and ${\cal B}$ corresponding to the vector equation (\ref{P21}), that is, from the matrices (\ref{P2AB}) and (\ref{P12AB}) of the size $(1,n)\times(1,n)$. First, we write down the equations
\[
 {\bf\Psi}'={\cal B}{\bf\Psi},\quad
 {\bf\Psi}_\zeta={\cal A}{\bf\Psi},\quad 
 {\bf\Psi}=\begin{pmatrix} \psi \\ \Psi \end{pmatrix},\quad \psi\in\mathbb{R},\quad \Psi\in\mathbb{R}^n
\]
(assuming that $A=0$). This gives the following systems of linear equations with respect to the scalar function $\psi$ and the vector function $\Psi$:
\begin{align*}
 &\left\{\begin{aligned}
  \psi'&= (1-\nu)\zeta\psi+({\bf u},\Psi),\\
  \Psi'&= \zeta\Psi+\psi {\bf u},
  \end{aligned}\right.\\
 &\left\{\begin{aligned} 
  \psi_\zeta&= \bigl(\nu^2(\nu-1)\zeta^2-(\nu-1)z-\nu p\bigr)\psi\\
    &\qquad +\zeta^{-1}\bigl(2p+z-\nu^2\zeta^2\bigr)(u,\Psi)+\nu({\bf u}',\Psi)-\zeta^{-1}({\bf u}'',\Psi),\\ 
  \Psi_\zeta&= \psi\bigl(\zeta^{-1}(2p+z-\nu^2\zeta^2){\bf u}-\nu {\bf u}'-\zeta^{-1}{\bf u}''\bigr)\\
    &\qquad +({\bf u},\Psi)\bigl(\nu {\bf u}+\zeta^{-1}{\bf u}'\bigr)
      -\zeta^{-1}({\bf u}',\Psi){\bf u} +\bigl(z-\nu^2\zeta^2\bigr)\Psi,
  \end{aligned}\right.
\end{align*}
where we denote $p=({\bf u},{\bf u})$, as before. Let
\[
 \varphi_0=({\bf u},\Psi),\quad \varphi_1=({\bf u}',\Psi),\quad \varphi_2=({\bf u}'',\Psi).
\]
Straightforward calculation leads to closed systems for $\psi$ and these quantities. The system with respect to $z$ is of the form
\[
 \left\{\begin{aligned}
  \psi'&= (1-\nu)\zeta\psi+\varphi_0,\\
  \varphi'_0&= p\psi+\zeta\varphi_0+\varphi_1,\\
  \varphi'_1&= \tfrac{1}{2}p'\psi+\zeta\varphi_1+\varphi_2,\\
  \varphi'_2&= (\tfrac{1}{2}p''-q)\psi+(\tfrac{3}{2}p'+1)\varphi_0+(3p+z)\varphi_1+\zeta\varphi_2,
 \end{aligned}\right.
\]
and a more cumbersome system for derivatives of the vector $(\psi,\varphi_0,\varphi_1,\varphi_2)$ with respect to $\zeta$ is obtained in similar way. Writing these systems in the matrix form, we arrive at the matrices (\ref{AB-P21-pq}).

The matrices ${\cal A}_{-1}=\operatorname*{res}_{\zeta=0}{\cal A}$ and $\tilde{\cal B}={\cal B}|_{\zeta=0}$ satisfy the degenerate Lax equation ${\cal A}'_{-1}=[\tilde{\cal B},{\cal A}_{-1}]$ without spectral parameter; from here it follows that the coefficients of the characteristic polynomial  
\[
 \det({\cal A}_{-1}-\mu I_4)=\mu^4-\frac{\mu^2}{2}H_1+\frac{1}{16}H_2
\]
are first integrals of (\ref{P21-pq}). It is easy to check that they have the form given in the statement.
\end{proof}

\begin{proposition} 
i) System (\ref{P22-pq}) admits an isomonodromic Lax representation (\ref{AB-pair}), where  
\begin{gather*}
 {\cal A}={\cal A}_0+\frac{1}{\zeta}{\cal A}_{-1},\qquad
 {\cal A}_{-1}=
 \begin{pmatrix}
  0 & 0 & p+z & 0 & -1 \\
  0 & 0 & 0 & 0 & 0 \\
  -p''+2p(p+z)+2q & 0 & p' & -2p & 0 \\
  p'(p+z)-q' & 0 & 2q & -p' & 0 \\
  p''(p+z)+p'+4pq-q'' & 0 & q' & 2q-p'' & 0 
 \end{pmatrix},\\[3mm]
 {\cal A}_0=
 \begin{pmatrix}
  0 & p+z-\zeta^2 & -\zeta & 0 & 0 \\
  p+z-\zeta^2 & 0 & 0 & -1 & 0 \\
 -2p\zeta & p' & 0 & 0 & 0 \\
 -p'\zeta & 2q & 0 & 0 & 0 \\
 (2q-p'')\zeta & q' & 0 & 0 & 0 
 \end{pmatrix},\qquad
 {\cal B}=
 \begin{pmatrix}
  0 & \zeta  & 1 & 0 & 0 \\
  \zeta & 0 & 0 & 0 & 0 \\
  2p & 0 & 0 & 1 & 0 \\
  p' & 0 & 0 & 0 & 1 \\
  p''-2q & 0 & 1 & 3p+z & 0
 \end{pmatrix}.
\end{gather*}
ii) The system admits the first integrals 
\begin{gather*}
H_1=q''-6pq+2zq-2(p+z)p''+p'^2-p'+2p(p+z)^2,\\
H_2= 2q(p''-2q)^2-2q'(p'p''-2p'q-pq')+(p'^2-4pq)(q''-p'-6pq-2zq).
\end{gather*}
\end{proposition}

\begin{proof}
The isomonodromic Lax representation for Eqs. (\ref{P22-pq}) is derived from Eqs. (\ref{P2AB}) and (\ref{P22AB}) in the same way as it was done in the proof of the previous Proposition. 

As before, the matrices ${\cal A}_{-1}$ and $\tilde{\cal B}={\cal B}|_{\zeta=0}$ satisfy the Lax equation ${\cal A}'_{-1}=[\tilde{\cal B},{\cal A}_{-1}]$ without spectral parameter, and the coefficients of the characteristic polynomial
\[
\det({\cal A}_{-1}-\mu I_5)=-\mu^5+\mu^3H_1-\mu H_2
\]
provide the first integrals of the system.
\end{proof}

Apparently, using the existence of pairs of integrals, the systems \eqref{P21-pq} and \eqref{P22-pq}  can be reduced to fourth-order equations contained in Cosgrove's classification \cite{Cosgrove_2000a}. However, this requires separate non-trivial calculations and is beyond the scope of this paper. 

\section{Vector KdV equation and generalizations of H\'enon–Heiles systems}\label{s:KdV}

In both third-order systems (\ref{NLS-1-t3}) and (\ref{NLS-2-t3}), the reduction ${\bf v}=-{\bf c}=\operatorname{const}$ is possible, which leads to two vector generalizations of the KdV equation:
\begin{equation}\label{KdVc-1}
 {\bf u}_t={\bf u}_{xxx}-3({\bf c},{\bf u}){\bf u}_x-3({\bf c},{\bf u}_x){\bf u} 
\end{equation}
and
\begin{equation}\label{KdVc}
 {\bf u}_t={\bf u}_{xxx}-3({\bf c},{\bf u}){\bf u}_x
  -3({\bf c},{\bf u}_x){\bf u}+3({\bf u},{\bf u}_x){\bf c}.    
\end{equation}
Here it is convenient to assume that the vectors ${\bf c}$ and ${\bf u}$ belong to ${\mathbb R}^{n+1}$, and not to ${\mathbb R}^n$, as usual.

Taking into account orthogonal transformations, the constant vector ${\bf c}$ can be set to ${\bf c}=(1,0,\dots, 0)$ without losing generality. If we then denote ${\bf u}=(u,v_1,\dots,v_n)$, then each of Eqs. (\ref{KdVc-1}) and (\ref{KdVc}) turns into a system for the scalar variable $u$ and the vector variable ${\bf v}$. The first system takes the form
\[
 u_t=u_{xxx}-6uu_x,\qquad 
 {\bf v}_t={\bf v}_{xxx}-3u{\bf v}_x-3u_x{\bf v},\qquad 
 u\in{\mathbb R},~~{\bf v}\in{\mathbb R}^n,
\]
that is, the usual KdV equation for $u$ is detached and linear equations for ${\bf v}$ arise on its background. This is not very interesting, and we will no longer consider this system. The second equation takes the form
\begin{equation}\label{KdV}
 u_t=u_{xxx}-3uu_x+3({\bf v},{\bf v}_x),\qquad 
 {\bf v}_t={\bf v}_{xxx}-3u{\bf v}_x-3u_x{\bf v},\qquad 
 u\in{\mathbb R},~~{\bf v}\in{\mathbb R}^n,
\end{equation}
where the scalar and vector components interact in a non-trivial way. It is this system that we will call the vector KdV equation. A zero curvature representation for it can be written in block matrices of size $(1,1,n,1)\times(1,1,n,1)$:
\begin{equation}\label{KdV-UV}
\begin{gathered}
{\cal V}=\lambda^2{\cal U}+\lambda{\cal U}_1+{\cal U}_0,\quad
{\cal U}=\begin{pmatrix}
 -\lambda & 1 & 0 & 0 \\ 
  u & 0 & 0 & 1\\ 
  {\bf v} & 0 & 0 & 0\\ 
  0 & u & {\bf v}^T & \lambda    
 \end{pmatrix},\quad
 {\cal U}_1=\begin{pmatrix}
  u & 0 & 0 & 0 \\
  u_x & 0 & {\bf v}^T & 0 \\ 
  {\bf v}_x & -{\bf v} & 0 & 0 \\ 
  0 & u_x & {\bf v}^T_x & -u 
 \end{pmatrix},\\
 {\cal U}_0=\begin{pmatrix}
  u_x & -u & {\bf v}^T & 0\\ 
  u_{xx}-u^2+({\bf v},{\bf v}) & 0 & {\bf v}^T_x & -u\\ 
  {\bf v}_{xx}-2u{\bf v} & -{\bf v}_x & 0 & {\bf v} \\   
  0 & u_{xx}-u^2+({\bf v},{\bf v}) & {\bf v}^T_{xx}-2u{\bf v}^T & -u_x  
 \end{pmatrix}. 
\end{gathered}
\end{equation}

\paragraph{Traveling wave solutions.} It is easy to check that substitutions
\[
 u(x,t)= u(z),\quad {\bf v}(x,t)=e^{tA}{\bf v}(z),\quad z=x-\gamma t,\quad A=-A^T
\]
reduce (\ref{KdV}) to equations
\[
 u'''-3 u u'+3({\bf v},{\bf v}')+\gamma u'=0,\qquad 
 {\bf v}'''-3u{\bf v}'-3u'{\bf v}+\gamma {\bf v}'-A{\bf v}=0.
\]
The change $u\mapsto u+\gamma/3$ allows us to put $\gamma=0$, and we arrive at the system
\begin{equation}\label{KdV-wave}
 u''=\frac{3}{2}u^2-\frac{3}{2}({\bf v},{\bf v})+\beta,\qquad 
 {\bf v}'''=3u{\bf v}'+3u'{\bf v}+A{\bf v}, \qquad u\in{\mathbb R},~~{\bf v}\in{\mathbb R}^n,
\end{equation}
which can be regarded as a generalization of the H\'enon–Heiles system \eqref{scalhh} with $\sigma=\frac{1}{2}$. The reason is that if $A=0$ then the second equation integrates and we obtain
\begin{equation}\label{zeroA}
 u''=\frac{3}{2}u^2-\frac{3}{2}({\bf v},{\bf v})+\beta,\qquad 
 {\bf v}''=3u{\bf v} +{\bf b}.
\end{equation}
We are not aware of any studies of this system for ${\bf b}\ne 0$ in the literature, even in the scalar case.

\paragraph{Galilean reduction.} The vector KdV equation (\ref{KdVc}) admits the one-parameter group of Galilean transformations of the form 
\[
 {\bf u}\mapsto {\bf u}+\tau {\bf c},\qquad x\mapsto x-3\tau({\bf c},{\bf c})t.
\]
By passing to its invariants and taking into account orthogonal transformations which preserve ${\bf c}$, we obtain the substitution
\[
 {\bf u}(x,t)=e^{tB}{\bf u}(z)+t{\bf c},\quad z=x-\frac{3}{2}({\bf c},{\bf c})t^2,\quad B+B^T=0,\quad B{\bf c}=0,
\]
which reduces (\ref{KdVc}) to the equation
\begin{equation}\label{KdVc-Gal}
 {\bf u}'''=3({\bf c},{\bf u}){\bf u}'+3({\bf c},{\bf u}'){\bf u}-3({\bf u},{\bf u}'){\bf c}+B{\bf u}+{\bf c},\qquad 
 {\bf u},{\bf c}\in{\mathbb R}^{n+1}. 
\end{equation}
For the KdV equation in the mixed scalar-vector form (\ref{KdV}), this substitution takes the form
\[
 u(x,t)=u(z)+t,\quad {\bf v}(x,t)=e^{tA}{\bf v}(z),\quad z=x-\frac{3}{2}t^2,
\]
and the reduced system takes the form (after integrating the equation for the scalar component)
\begin{equation}\label{KdV-Gal}
 u''=\frac{3}{2}u^2-\frac{3}{2}({\bf v},{\bf v})+z,\quad 
 {\bf v}'''=3u{\bf v}'+3u'{\bf v}+A{\bf v},\quad 
 u\in{\mathbb R},~~ {\bf v}\in{\mathbb R}^n,~~ A+A^T=0.
\end{equation}
This system can be regarded as a kind of vector generalization of the \ref{P1} equation which appears (up to a scaling) if ${\bf v}=0$.

If we do not use orthogonal transformations, that is, consider zero matrices $B$ and $A$, then the equation for the vector component is also integrated and we arrive at the equations
\begin{equation}\label{KdVc-Gal0}
 {\bf u}''=3({\bf c},{\bf u}){\bf u}-\frac{3}{2}({\bf u},{\bf u}){\bf c}+z{\bf c}+{\bf b},\qquad 
 {\bf u},{\bf c},{\bf b}\in{\mathbb R}^{n+1}
\end{equation}
or, in the scalar-vector form,
\begin{equation}\label{KdV-Gal0}
 u''=\frac{3}{2}u^2-\frac{3}{2}({\bf v},{\bf v})+z,\qquad 
 {\bf v}''=3u{\bf v}+{\bf b},\qquad u\in{\mathbb R},~~ {\bf v},{\bf b}\in{\mathbb R}^n.
\end{equation}
Notice that system (\ref{KdV-Gal0}) is Hamiltonian with respect to brackets
\[
 \{u,u'\}=1,\quad \{v_i,v'_i\}=-1
\]
and the Hamiltonian
\[
 H=\frac{1}{2}(u')^2-\frac{1}{2}({\bf v}',{\bf v}')-\frac{1}{2}u^3
   +\frac{3}{2}u({\bf v},{\bf v})+({\bf b},{\bf v})-zu.
\]
The remaining orthogonal transformations make possible to bring ${\bf b}$ either to the form ${\bf b}=0$ or ${\bf b}=(1,0,\dots,0)$. In the first case, the system (\ref{KdV-Gal0}) is reduced to two second-order equations for $u$ and $g=({\bf v},{\bf v})$ (cf. with Eqs. (3.46) and (3.47) from \cite{Cosgrove_2000a})
\begin{equation}\label{KdV-Gal00}
 u''=\frac{3}{2}u^2-\frac{3}{2}g+z,\qquad 2gg''-(g')^2=12ug^2+\gamma.
\end{equation}
In the second case, we again detach the first coordinate of the vector: let ${\bf v}=(v_1,v_2,\dots,v_n)$ and $g=\sum_{i>1}v^2_i$, then the system appears
\begin{equation}\label{KdV-Gal01}
 u''=\frac{3}{2}u^2-\frac{3}{2}(v^2_1+g)+z,\quad 
 v''_1=3u v_1+1,\quad 
 2gg''-(g')^2=12ug^2+\gamma.
\end{equation}

\begin{remark} If $n=3$, system (\ref{KdV-Gal0}) with  ${\bf v}=(v_1,v_2,v_3)$, ${\bf b}=(b_1,b_2,b_3)$ coincides with the matrix P$_1$ equation \cite{Balandin_Sokolov_1998}
\[
 U''= \frac{3}{2}U^2+zI_2+B, 
\]
where 
\[
 U=\begin{pmatrix}
 u+iv_1 & -v_2+iv_3\\
 v_2+iv_3 & u-iv_1
 \end{pmatrix},\qquad
 B=\begin{pmatrix}
  ib_1 & -b_2+ib_3\\
  b_2+ib_3 & -ib_1
 \end{pmatrix},
\]
and (\ref{zeroA}) is an autonomous version of this equation.

\end{remark}

\paragraph{Self-similar reduction.} The substitution
\[
 u(x,t)= (\varepsilon\tau)^2u(z),\quad 
 {\bf v}(x,t)= (\varepsilon\tau)^2e^{\log(\tau)A}{\bf v}(z),\quad 
 \tau=t^{-1/3},\quad z=\varepsilon\tau x,\quad 3\varepsilon^3=-1,
\]
resembles to substitution (\ref{self}), but with different weights of variables. It reduces (\ref{KdV}) to the system
\begin{equation}\label{KdV-self}
 u'''=3uu'-3({\bf v},{\bf v}')+zu'+2u,\qquad 
 {\bf v}'''=3u{\bf v}'+3u'{\bf v}+z{\bf v}'+(A+2){\bf v}
\end{equation}
where $u\in{\mathbb R}$, ${\bf v}\in{\mathbb R}^n$ and $A+A^T=0$.

\paragraph{Lax representations.} We present Lax representations for the obtained reductions inherited from the zero curvature representation with matrices (\ref{KdV-UV}). All matrices are of block size $(1,1,n,1)\times (1,1,n,1)$. Let
\[
 {\cal B}_0=\begin{pmatrix}
  -1 & 0 & 0 & 0\\   
   0 & 0 & 0 & 0\\ 
   0 & 0 & 0 & 0\\   
   0 & 0 & 0 & 1 
 \end{pmatrix},\quad
 {\cal B}_1=\begin{pmatrix}
  0 & 1 & 0 & 0\\   
  u & 0 & 0 & 1\\ 
  {\bf v} & 0 & 0 & 0\\   
  0 & u & {\bf v}^T & 0 
 \end{pmatrix},\quad
 {\cal B}_2=
 \begin{pmatrix}
 u & 0 & 0 & 0\\   
 u' & 0 & {\bf v}^T & 0\\ 
 {\bf v}' & -{\bf v} & 0 & 0\\   
 0 & u' & ({\bf v}')^T & -u 
 \end{pmatrix}. 
\]
The matrix ${\cal B}$ in all three cases is ${\cal B}=\zeta{\cal B}_0+{\cal B}_1$. The matrices ${\cal A}$ are specified in the following statement.

\begin{proposition}
The traveling wave equations (\ref{KdV-wave}) admit the isospectral Lax pair ${\cal A}_z=[{\cal B},{\cal A}]$ with ${\cal A}=\zeta^3{\cal B}_0+\zeta^2{\cal B}_1+\zeta{\cal B}_2 +{\cal B}_3$, where
\[
 {\cal B}_3=
 \begin{pmatrix}
  u' & -u & {\bf v}^T & 0\\   
  \frac{1}{2}u^2-\frac{1}{2}({\bf v},{\bf v})+\beta & 0 & ({\bf v}')^T & -u\\ 
  {\bf v}''-2u{\bf v} & -{\bf v}' & -A & {\bf v}\\   
  0 & \frac{1}{2}u^2-\frac{1}{2}({\bf v},{\bf v})+\beta & ({\bf v}'')^T-2u{\bf v}^T & -u' 
 \end{pmatrix};
\]
the Galilean reduction (\ref{KdV-Gal}) admits the isomonodromic Lax pair (\ref{AB-pair}) with ${\cal A}=\zeta^4{\cal B}_0+\zeta^3{\cal B}_1+\zeta^2{\cal B}_2 +\zeta{\cal B}_3+{\cal B}_4$, where
\[
 {\cal B}_3=
 \begin{pmatrix}
  u' & -u & {\bf v}^T & 0\\   
  \frac{1}{2}u^2-\frac{1}{2}({\bf v},{\bf v})+z & 0 & ({\bf v}')^T & -u\\ 
  {\bf v}''-2u{\bf v} & -{\bf v}' & -A & {\bf v}\\   
  0 & \frac{1}{2}u^2-\frac{1}{2}({\bf v},{\bf v})+z & ({\bf v}'')^T-2u{\bf v}^T & -u' 
 \end{pmatrix},\quad
 {\cal B}_4=\begin{pmatrix}
   0 & 0 & 0 & 0\\   
   1 & 0 & 0 & 0\\ 
   0 & 0 & 0 & 0\\   
   0 & 1 & 0 & 0 
 \end{pmatrix};
 \]
the self-similar reduction (\ref{KdV-self}) admits the isomonodromous Lax pair (\ref{AB-pair}) with ${\cal A}=-(\zeta-z\zeta^{-1}){\cal B}-{\cal B}_2-\zeta^{-1}{\cal B}_3$, where
\[
 {\cal B}_3=
 \begin{pmatrix}
 u'+1 & -u & {\bf v}^T & 0\\   
 u''-u^2+({\bf v},{\bf v}) & 0 & ({\bf v}')^T & -u\\ 
 {\bf v}''-2u{\bf v} & -{\bf v}' & -A & {\bf v}\\   
 0 & u''-u^2+({\bf v},{\bf v}) & ({\bf v}'')^T-2u{\bf v}^T & -u'-1 
 \end{pmatrix}.
\]
\end{proposition}

\section{Conclusion}

In papers \cite{Sokolov-Wolf}, \cite{Meshkov-Sokolov} and \cite{Balakhnev-Meshkov}, a general theory of integrable vector evolution equations of the form 
\begin{equation}\label{evol}
{\bf u}_t = \sum f_i {\bf u}_i, \qquad {\bf u}_i = \frac{\partial^i {\bf u}}{\partial x^i}, \qquad   {\bf u}\in{\mathbb R}^n
\end{equation}
was developed and many examples of such equations have been constructed. It is assumed that the coefficients $f_i({\bf u}, {\bf u}_x, \dots)$ in the formula \eqref{evol} are scalar. In the simplest isotropic case, they depend on the scalar products of their arguments (see, for example, \eqref{mKdV-1} and \eqref{mKdV-2}). Using various similarity reductions of these evolution equations, a large number of vector ODE systems can be constructed, both Liouville integrable and satisfying the Painlev\'e property. The systematic study of these systems is an open task awaiting its researcher. For instance, it would be useful to study self-similar reductions of integrable vector systems of the derivative NLS type \cite{Sokolov-Wolf}.

Note that the reductions leading to isotropic vector system of the form
\begin{equation}\label{ODE}
{\bf u}_k = \sum_{i=0}^{k-1} g_i {\bf u}_i , \qquad {\bf u}\in{\mathbb R}^n
\end{equation}
with scalar coefficients $g_i$ are somehow degenerate in the sense that, without loss of generality, one can assume that $n\le k$. Indeed, the solution of this system is determined by the initial data ${\bf u}_0(0)={\bf r}_0,\dots,{\bf u}_{k-1}(0)={\bf r}_{k-1}$, where ${\bf r}_i$ are arbitrary constant vectors in the $n$-dimensional space.

\begin{lemma}\label{lem} 
The vector ${\bf u}(x)$ lies in the subspace spanned by ${\bf r}_0,\dots,{\bf r}_{k-1}$ for all $x$, regardless of the dimension $n$ of the entire space.
\end{lemma}
\begin{proof} 
Equations $({\bf u},{\bf C})=({\bf u}_1,{\bf C})=\dots=({\bf u}_{k-1},{\bf C})=0$ define an invariant submanifold of equation (\ref{ODE}), for any constant vector ${\bf C}$.
\end{proof}

By this reason, if we wish to construct vector ODE reductions of arbitrary dimension, we should generalize the class of systems of the form \eqref{ODE}, by relaxing the condition that coefficients $g_i$ are scalar. This happens, for example, when the symmetry group used for reduction contains a subgroup of orthogonal transformations and an arbitrary skew-symmetric matrix appears in the coefficients (see, for example, the systems \eqref{KSGarnier}, \eqref{KSGarnier-2}, \eqref{P21}, \eqref{P22} and \eqref{KdV-Gal}). According to Lemma \ref{lem}, fully isotropic systems  \eqref{ODE} are interesting only for small $n$. Further reducing of their order is possible by restriction to the invariants of the symmetry subgroup ${\rm O}(n)$. Some examples were discussed in section \ref{s:A0}.

One of non-trivial problems that we have not touched upon in this article is the construction of B\"acklund transformations for integrable vector systems of ODEs. For the systems considered in this article, it is solved only for the autonomous Garnier type systems \cite{Suris_2003}, and for the non-autonomous deformation (\ref{Garnier-1}) \cite{Adler_Kolesnikov_2023}. It is known that B\"acklund transformations are associated with various algebraic structures such as the Yang--Baxter equation and others. It would be interesting to develop a theory of such structures in the vector case. 

\subsubsection*{CRediT authorship contribution statement}

V.E. Adler: Conceptualization (equal); Investigation (equal); Writing --- original draft (equal); Writing --- review \& editing (equal). V.V. Sokolov: Conceptualization (equal); Investigation (equal); Writing --- original draft (equal); Writing --- review \& editing (equal).

\subsubsection*{Declaration of competing interest}

The authors have no conflicts to disclose.

\subsubsection*{Acknowledgments}

The authors are grateful to R.~Conte for helpful discussions. The research of VVS is an output of a research project implemented as part of the Basic Research Program at the National Research University Higher School of Economics (HSE University). The research of VEA is supported by state contract FFWR-2024-0012 of the Landau Institute for Theoretical Physics of the Russian Academy of Sciences.

\subsubsection*{Data availability}

No data was used for the research described in the article..

 
\end{document}